\documentclass[a4paper,10pt]{article}

\usepackage[english]{babel}
\usepackage[T1]{fontenc}
\usepackage[utf8]{inputenc}
\usepackage{units}
\usepackage[margin=1.25in]{geometry}

\usepackage{amsmath}
\usepackage{afterpage}

\usepackage{graphicx}
\usepackage{subfigure}
\usepackage{wrapfig}
\usepackage{indentfirst}
\usepackage{latexsym}
\usepackage{sidecap}
\usepackage{booktabs}
\usepackage{amsthm}
\usepackage{amssymb}

\usepackage{setspace}

\usepackage{mathrsfs}
\usepackage{textcomp}
\usepackage{braket}
\usepackage{bbm}
\usepackage{colortbl}
\usepackage{bbm}
\usepackage{comment}
\usepackage[colorinlistoftodos]{todonotes}
\usepackage{breakcites}

\theoremstyle{plain}
\newtheorem{Theorem}{Theorem}[section]
\theoremstyle{definition}
\newtheorem{Definition}{Definition}[Theorem]
\theoremstyle{remark}
\newtheorem{Remark}[Theorem]{Remark}
\theoremstyle{remark}

\theoremstyle{plain}
\newtheorem{Lemma}[Theorem]{Lemma}
\theoremstyle{plain}

\theoremstyle{plain}

\theoremstyle{remark}

\theoremstyle{remark}

\theoremstyle{remark}

\theoremstyle{remark}

\renewcommand\vec[1]{\overrightarrow{#1}}
\newcommand\cev[1]{\overleftarrow{#1}}
\newcommand\RR{\mathbb{R}}
\newcommand\NN{\mathbb{N}}

\usepackage{comment}
\usepackage[colorinlistoftodos]{todonotes}

\newcommand{\Ind}[1]{\mathbbm{1}_{\left [#1\right ]}}

\newcommand{\GG}{\mathbb{G}}

\newcommand{\Ver}{\mathrm{v}}

\usepackage{fancyhdr}
\usepackage{calc} 
\begin{document}

\date{}

\title{A lending scheme for a system of interconnected banks with probabilistic constraints of failure}

\author{Francesco Cordoni$^1$, Luca Di Persio$^2$ and Luca Prezioso$^3$\\
$^2$ University of Verona - Department of Computer Science\\ Strada le Grazie, 15 Verona, 37134-ITALY\\
e-mail: francescogiuseppe.cordoni@univr.it
\\[2pt]
$^2$ University of Verona - Department of Computer Science\\ Strada le Grazie, 15 Verona, 37134-ITALY\\
e-mail: luca.dipersio@univr.it
\\[2pt]
$^2$ University of Trento - Department of Mathematics \\ Via Sommarive, 14 Trento, 38123-ITALY\\
e-mail: luca.prezioso@unitn.it
\\[2pt]
}

\maketitle

\begin{abstract}
We derive a closed form solution for an optimal control problem related to an interbank  lending schemes subject to terminal probability constraints on the failure of banks which are interconnected through a financial network. The derived solution applies to a real banks network by obtaining a general solution when the aforementioned probability constraints are assumed for all the banks. We also present a direct method to compute the systemic relevance parameter for each bank within the network. 
\end{abstract}


\section{Introduction}
One of the most relevant changes within the financial world has been caused by the worldwide crisis of 2007-2008. Starting from that breaking event, financial analysts, bank practitioners, applied mathematicians and economists, have been pushed to rethink the models they were used to work with. In particular, it was necessary to stop relying on a series of assumptions turned out to be too far from real markets, as well as from the {\it new} altered financial worldwide scenario and its changed functioning. 
Such a {\it big crunch}, along with its consequences,  forced both investors and financial institutions to be aware that almost every financial quantity is exposed to concrete failure risk.\\
As an example, the standard Black and Scholes (BS) model, whose restrictions on coefficients have been the focus of several studies determining a plethora of alternative and effective approaches, see, e.g., \cite{Mer1}, has shown evident limits. In particular, as in the BS model, we base our framework on the geometric Brownian motion. Nevertheless, a major focus of our work will be on default probabilities that any financial entity must face.\\
We underline that such {\it credit risk} analysis has seen an increasing interest in the theoretical financial community, pushing  the development of mathematically rigorous models  which take into account both the {\it risk exposure} factor and related  {\it default events}.
Along aforementioned lines, two main approaches have
been developed: the totally unexpected failure method, also known as {\it reduced-form intensity-based model}, see, e.g., ~\cite[ch.~8]{BR}, and the triggered failure method, also known as {\it structural model}, see, e.g., ~\cite[sec.~1.4]{BR}. 
Mathematically speaking, the first approach defines the default time as the first jump time of some stochastic process, so that the default event is completely inaccessible to the probabilistic reference filtration modeling the information flow available to traders.
After exogenously specifying the conditional probability of default, a typical method to deal with this inaccessibility issue
is based on filtration enlargement, see, e.g., \cite{BR,Lan,Rog}.\\
The second approach supposes the default event to be triggered as soon as the value of the financial entity reaches an endogenous lower threshold. Hence, one of the main issues of the method is the evolution modeling of both the financial entity value and of its capital structure. Therefore, differently from the first mentioned approach, the default time results in being a predictable stopping time with respect to the reference filtration. Let us recall that structural default risk models have been extensively studied in literature, see, e.g.,~\cite{BR,Bla,Mer,Mug}.
In what follows we focus our attention on the latter approach, taking into consideration a network of interconnected financial entities, such as banks or general economic agents, who are willing to lend money to each other. 
We assume that the {\it bankruptcy event } for a bank occurs when its capital hits a lower barrier whose value is linked to the characterization of the whole system.
As a main reference setting we refer to the one introduced in \cite{Eis}, then generalized in \cite{Cap,Lip,Rog}. In particular, following \cite{Cap}, we consider a {\it financial supervisor}, usually referred as {\it lender of last resort} (LOLR) aiming at guarantee the {\it wellness} of the financial network, by lending money to those agents who are near to default. At the same time, the LOLR also tries to minimize a given cost function.\\
Our results also allows  
to compute the optimal controls for highly complex networks, as the real banking ones. The main novelty of our solution is that, in addition to considering a LOLR who lends money aiming at minimizing a given cost function, we further  assume fixed probability constraints the banks have to satisfy at a specific terminal time. From a financial point of view, such constraint implies that the LOLR optimal strategy has to be derived satisfying the  assumption that each bank is characterized by a probability of bankruptcy. As in \cite{Mer}, we assume that a bank may fail only at a fixed terminal time, namely it goes under bankruptcy if, at terminal time,  its wealth is below a given threshold.
This allows us to derive the optimal strategy exploiting techniques related to stochastic target problems. We recall that 
first results in this direction have been derived in \cite{Son}, where an {\it ad hoc} dynamic programming principle has been provided. Later, several papers appeared generalizing such results by considering different constraints schemes, spanning from expectation constraints at fixed time, to almost sure constraints, see, e.g., \cite{Bok,Bou,Bou2,Bou3,Gru}. In \cite{Ozg}, an optimal solution is derived within a similar setting, but without using the stochastic target problem approach. Since in the above mentioned papers examples of  concrete solutions are often missing,
at the end of this work we consider an example. In particular, we compare our result with the one obtained in \cite{Cap} 
limiting, for the sake of clarity,  ourselves to a small set of interconnected banks, the case of larger network being of easy derivation. 
Moreover, because  the model construction is strongly based on the mathematical theory of networks, we will exploit its characteristics in order to derive a  \textit{page rank} approach, first introduced in \cite{Page}, which will be used to determine the {\it relative importance} of any bank in the network. We then exploit this quantity to decide the admitted probability of each bank's failure, requiring that {\it important banks} 
have larger {\it non-failure probability}, hence adopting a {\it too big to fail paradigma} .\\
The present work is organized as follows: in Section \ref{SEC:GS} we introduce the main setting, giving the  mathematical and financial definitions; in Section \ref{SEC:PC} we  introduce the optimal control problem with probability constraints and we provide its solution;  in Section \ref{Appl} we present the Pagerank method for the relative importance of the banks in the network and we apply the derived results to a toy example.

\section{The general setting}\label{SEC:GS}

Following the financial network setting proposed in \cite{Eis,Rog},  see appendix \ref{G:framework} for further details,
we consider a network composed by $n$ nodes, each of them representing a different financial agent, and we denote by $X^i(t)$ the asset value of the $i^\text{th}$ agent at time $t \in [0,T]$, being $T<\infty$ a fixed positive terminal time. 
Each node may have nominal liabilities to other nodes directly connected with it. In this case, we denote by $L_{i,j}(t)$ the payment that the bank $i$ owes to the bank $j$, at time $t \in [0,T]$. Then, we introduce the time-dependent \textit{liabilities matrix} $\mathcal{L}(t) = \left (L_{i,j}(t)\right )_{n \times n}$, defined as
\begin{equation}\label{liability}
\begin{cases}
L_{i,j} (t) & \iota^+_{i,j} \not = 0\, ,\\
0 & \mbox{ otherwise }\, ,
\end{cases}
\end{equation}
where, as shown in appendix \ref{G:framework}, $\iota^+_{i,j}$ is equal to one if  $i$ and $j$ are connected, while it equals zero otherwise. In particular, equation \eqref{liability} explicitly states  that there cannot be any cash flow between any two banks which are not edge-connected.\\
At any time $t \in [0,T]$ , the $i^\text{th}$ agent may also have an exogenous cash inflow $F^i(t) \geq 0$.\\
We will denote by $u_i(t)$ the payment made at time $t \in [0,T]$ by the $i^\text{th}$ bank, whereas $\bar{u}_i(t) = \sum_{j=1}^n L_{i,j}(t)$ is the \textit{total nominal obligation} of node $i$ towards all other nodes. Therefore, if $\bar{u}_i(t) = u_i(t)$, then $i$ has satisfied all its liabilities. We also introduce the \textit{relative liabilities matrix} $\Pi (t) = \left (\pi_{i,j}(t\right ))$ defined as
\[
\begin{cases}
\frac{L_{i,j} (t)}{\bar{u}_i(t)} & \bar{u}_i(t) > 0\, ,\\
0 & \mbox{ otherwise }\, .
\end{cases}
\]
Let us notice that the matrix $\Pi (t)$ is row stochastic, in the sense that $\sum_{j=1}^n \pi_{i,j} (t) = 1$, so that $\pi_{i,j}(t)$ represents the proportion of the total debt at time $t$ that the node $i$ owes to the node $j$.\\
Similarly, we can define the cash inflow of the node $i$ as the sum of the exogenous cash inflow $F^i(t)$ plus the total payment that node $i$ receives at time $t$ by other nodes, that is $\sum_{j=1}^n \pi_{i,j}^T(t)\, u_j(t)$ 
, where we denoted the transposed of the relative liabilities matrix and its elements as $\Pi^T=(\pi_{i,j}^T(t))$.
We thus have that the value of the $i^\text{th}$ node at time $t \in [0,T]$ is given by
\begin{equation}\label{EQN:ValB}
\bar{V}^i(t) =  \sum_{j=1}^n \pi_{i,j}^T(t)\, u_j(t) + F^i(t) - \bar{u}_i(t)\, .
\end{equation}
Let us now introduce the notion of \textit{clearing vector} as a specification of the payments made by each of the banks in the financial system, see, e.g.,  \cite[Definition 1]{Eis},   \cite[Definition 2.6]{Rog}.
In what follows, if not otherwise specified, we will use standard point-wise ordering for vectors in $\RR^n$, namely for every $x$, $y \in \RR^n$ it holds $x \leq y$ if and only of $x_i \leq y_i$, for any $i=1,\dots,n$.
\begin{Definition}\label{DEF:CV}
In the aforementioned financial setting, see also appendix \ref{G:framework}, a \textit{clearing vector} is a vector $u^* (t) \in [0,\bar{u}(t)]$ satisfying 
\begin{itemize}
\item \textbf{Limited liabilities:}
\[
u^*_i (t) \leq \sum_{j=1}^n \pi_{i,j}^T(t)\, u_j^*(t) + F^i(t)\, , \quad i = 1, \dots, n \, ;
\]
\item \textbf{Absolute priority:} that is either obligations are paid in full, or all value of the node is paid to creditors, i.e. 
\[
u^*_i (t) = \begin{cases}\bar{u}_i(t) \text{, if } \bar{u}_i(t) \leq  \sum_{j=1}^n \pi_{i,j}^T(t)\, u_j^*(t) + F^i(t)\\ \sum_{j=1}^n \pi_{i,j}^T(t)\, u_j^*(t) + F^i(t) \text{, otherwise}.\end{cases}
\]
\end{itemize}
\end{Definition}
Existence and uniqueness of a \textit{clearing vector}, in the sense of Definition \ref{DEF:CV},
is treated in \cite{Eis,Rog}. In particular, in \cite{Eis} it is shown that $u^*(t)$ is a clearing vector if and only if 
\begin{equation}\label{EQN:CV}
u^*(t) = \bar{u}_i(t) \wedge \left (\sum_{j=1}^n \pi_{i,j}^T(t)\, u_j^*(t) + F^i(t)\right )\, .
\end{equation}
Equation \eqref{EQN:CV} can be interpreted as follows: the term $\bar{u}_i(t)$ specifies which $i-node$ owes to the other nodes at time $t \in [0,T]$, whereas the second term $\left (\sum_{j=1}^n \left (\pi_{i,j}(t)\right )^T u_j^*(t) + F^i(t)\right )$ represents the cash inflow for the node $i$ at time $t \in [0,T]$. 
Consequently, \textit{clearing vector} represents the payment at time $t$ of each node:  each node pays the minimum between what it has and what it owes. Combining equation \eqref{EQN:ValB} and  \eqref{EQN:CV}, we say that the bank $i$ is in {\it default} if it is not able to meet all of its obligations, therefore
the value of a bank equals
\begin{equation}\label{EQN:ValBDef}
V^i(t) = \left ( \sum_{j=1}^n \pi_{i,j}^T(t)\, \bar{u}_j(t) + F^i(t) - \bar{u}_i(t)\right )^+\, ,
\end{equation}
where $(f(x))^+$ denotes the positive part of the function $f$, so that if $\bar{V}^i(t) \leq 0$, then the bank $i$ is in default, and we set its value to $V^i(t)=0$.\\
To simplify the notation, let us define the matrix 
\[
\tilde{L} = \left (\tilde{L}_{i,j}\right )_{n \times n} := L - \text{diag}( u(t)),
\]
where  diag$(u(t))$ indicates a $n\times n$ diagonal matrix with the vector $u(t):= (u_1(t),\dots, u_n(t))$ as its diagonal. The matrix $\tilde{L}$ has entry $L_{i,j}(t)$ off diagonal, and $-\sum_{j=1}^n L_{i,j}(t)$, representing the total payment that the bank $i$ owes at time $t$ to other nodes, on the main diagonal.\\
Following  \cite{Lip}, we assume the liabilities between banks to evolve according the following equation
\begin{equation}\label{EQN:Liab}
\frac{d}{dt} L_{i,j} (t) = \mu_{ij} L_{i,j} (t)\, ,
\end{equation}
for a fixed positive growth rate $\mu >0$.
We stress that the 
present setting can be generalized taking $L$ as a geometric Brownian motion. In such scenario the terminal constraint becomes stochastic. Nonetheless, computing the conditional expectation of the terminal constraint, it is possible to recover results analogous the the setting used in the present paper. We leave this topic to be addressed in a future work. Similarly, we assume the bank $i$, at any time $t$, invests the difference between cash inflow ans cash outflow in an exogenous asset $X^i(t)$ whose dynamic is given by
\[
\mathrm{d} X^i(t) = X^i(t) \left (\mu^i \,\mathrm{d}t + \sigma_i\, \mathrm{d}W^i(t)\right )\, , \quad i =1, \dots, n\, .
\]
Moreover, see \cite{Lip}, we introduce continuous (deterministic) default boundaries as follows
\[
X^i (t) \leq v^i(t)\,,\quad\mathbb{P}~a.s.\,,
\]
with
\begin{equation}\label{EQN:ContDet}
v^i(t) :=
\begin{cases}
R^i\left (\bar{u}_i(t)- \sum_{j=1}^n \pi_{i,j}^T(t)\, \bar{u}_j(t)\right ) & t< T\, ,\\
\bar{u}_i(t)- \sum_{j=1}^n \pi_{i,j}^T(t)\, \bar{u}_j(t) & t = T\, ,
\end{cases}
\end{equation}
where $R^i \in (0,1)$, $i = 1,\dots,n$, are suitable constants representing the \textit{recovery rate} of the bank $i$.

\section{The stochastic optimal control with probability constraints}\label{SEC:PC}

In what follows, we introduce the mathematical formulation of our problem, expressing it as an optimal control problem with terminal probability constraint. Furthermore, we provide an analytic solution which allows 
us
to compute the optimal controls.\\
In what follow we consider a complete filtered probability space $\left (\Omega,\mathcal{F},\left (\mathcal{F}_t\right )_{t \geq 0},\mathbb{P}\right )$ satisfying usual assumptions, namely right--continuity and saturation by $\mathbb{P}$--null sets. Next results will be applied in Section \ref{Appl} to analyse some  financial networks tou models.
As in the paper by Capponi et al. \cite{Cap}, we consider a financial supervisor, called \textit{Lender Of Last Resort} (LOLR), connected to any node belonging to the financial network. The LOLR aims at saving the network from default, ant it is  assumed to have {\it full information} about the network state. In particular, at any time $t$ the LOLR can lend money to the bank $i$, $i = 1,\dots,n$, so that the controlled evolution of the bank $i$ satisfies 
\begin{equation}\label{EQN:XTau}
\begin{split}
\mathrm{d}X^i_\alpha(t) &= \left (\mu^i\, X_\alpha^i(t)  + \alpha^i(t) \right ) \mathrm{d}t + \sigma_i\, X_\alpha^i(t)\,\mathrm{d}W^i(t)\, ,
\end{split}
\end{equation}
being $\alpha^i(t)$ the loan from the LOLR to the bank $i$, at time $t \in [0,T]$ and such that $\alpha \in \mathcal{A}$, where $\mathcal{A}$ is the collection of progressively measurable processes $\alpha \in L^2([0,T])$, $\mathbb{P}-a.s.\,$. In particular, $\alpha(t)$-vector components represent the amounts of money lent to each bank by the LOLR.\\
In order to derive a closed form solution, we will consider the setting proposed originally by Merton in \cite{Mer}. Therefore,  we assume that default can happen only at some fixed time $t_i$, $i=1,\dots,l$, $l<\infty$, hence allowing to only consider constraints defined at terminal time. This allows to avoid introducing  strong 
bonds at each time $t \in [0,T]$.\\
Let us note  that an analogous result can be obtained considering  banks allowed to fail at some discrete times $t_1<t_2<\dots<t_M =T$, by
separately considering any control problem between two fixed time $[t_i,t_{i+1}]$. This  allows to obtain a global control solution by  gluing together an
an ordered sequence of optimal control problems, then exploiting results presented along subsequent sections. Nevertheless, 
the obtained glued solution is not the optimal one. In fact, in deriving the optimal solution for any time $t$, one has to consider also possible future evolutions of the system. We shall study latter scenario 
in a future research exploiting the results here provided, hence deriving the global optimal solution via {\it backward induction}, as addressed in 
\cite{CDP,Pha}. Assuming  that the LOLR aims at minimizing lend resources implies that he tries to minimize the functional
\begin{equation}\label{EQN:CostFunc}
J(\alpha) = \mathbb{E}\biggl[\frac{1}{2} \sum_{i=1}^n\int_{0}^{T} \alpha^i(s)^2\, \mathrm{d}s\biggr]\, .
\end{equation}
Moreover, the LOLR minimizes equation \eqref{EQN:CostFunc} over the probabilistic constraint
\begin{equation}\label{EQN:ProbC}
\mathbb{P}\left (X^i(T) \geq v^i(T)\right ) \geq q^i\,, \quad i=1,\dots,n \, , 
\end{equation}
for  suitable constants $q^i \in (0,1)$, $i=1,\dots,n$. For the ease of notation, in what follows we will drop the index $i$. Hence, with respect to the agent $i$, we are going to solve the general control problem, then we apply such result to all banks in the system.\\
Therefore, let us consider the value of a bank evolving over time according to
\begin{equation}\label{EQN:BankPC}
\begin{split}
\mathrm{d}X^i (t) &= \left (\mu^i\, X^i(t)  + \alpha^i(t) \right ) \mathrm{d}t + \sigma^i\, X^i(t)\,\mathrm{d}W^i(t)\, ,\\
X^i(0)&=x^i\,,\quad i=1\,,\dots,n\,,
\end{split}
\end{equation}
and the corresponding default value $v^i:=v^i(T)$ at terminal time. Moreover, in what follows, we shall  assume that the external supervisor chooses the control $\alpha$ minimizing the following criterion
\begin{gather}
J(t,\alpha) = \mathbb{E}\biggl[\frac{1}{2} \sum_{i=1}^n\int_{t}^{T} \alpha_i(s)^2\, \mathrm{d}s\,\biggr|\;\mathcal{F}_t\biggr]\,,\label{EQN:CrirMin}\\
s.t.\quad \mathbb{P}\left (X^i(T) \geq v^i|\mathcal{F}_t\right ) \geq q^i\,,\quad i = 1,\,	\dots n\,.\tag{PC}\label{EQN:PC}
\end{gather}

\subsection{Reduction to a stochastic target problem}\label{SEC:ST}
In the current section we are going to formally introduce the \textit{Hamilton--Jacobi--Bellman} (HJB) equation associated to the  control problem defined in eq. \eqref{EQN:CostFunc}, subject to constraint given by eq. \eqref{EQN:ProbC}, hence reducing the related optimal control problem to a stochastic target one. We stress that in what follows, due to the structure of the optimal control problem, we will focus on single agent $i$. In particular, to avoid heavy notation, if not otherwise stated, we will denote for short $X:=X^i$.\\
Exploiting the value function form given by eq. \eqref{EQN:CrirMin} and by rewriting the terminal probability in equation \eqref{EQN:PC} as an expectation, namely
\[
\mathbb{P}\left (X(T) \geq v\,\big|\;\mathcal{F}_t\right )= \mathbb{E}\left [\Ind{X(T) \geq v}\,\big|\;\mathcal{F}_t\right ]\, ,
\]
then we have the following
\begin{Lemma}
Given the stochastic optimal control problem with terminal probability constraint \eqref{EQN:PC}, then the terminal probability constraints holds if and only if there exists an adapted sub-martingale $\left (P(s)\right )_{s \in [t,T]}$ such that
\[
P(t) = q\, ,\quad P(T) \leq \Ind{X(T) \geq v}\, .
\]
\end{Lemma}
\begin{proof}
Let us first prove $(\Leftarrow)$: since $P(s)$ is a sub-martingale we have that
\[
\mathbb{E}\left [\Ind{X(T) \geq v}\right ] \geq \mathbb{E}\left [P(T)\,|\mathcal{F}_t\right ] \geq P(t) =q\, . 
\]
To prove the converse implication $(\Rightarrow)$, let us first denote
\[
\begin{split}
q_0 &:= \mathbb{E}\left [\Ind{X^s(T) \geq v}\right ]\, ,\\
P(s) &:=  \mathbb{E}\left [\left .\Ind{X^s(T) \geq v}\right |\mathcal{F}_s\right ] - (q_0-q)\, ,
\end{split}
\]
where  $X^s$ represents the solution with initial time $s \in [t,T]$, then $P$ is an adapted martingale and the claim follows.
\end{proof}
Let us note that when the probability constraints is {\it active}, the sub-martingale $P$ is given by
\[
P(s) = \mathbb{E}\left [\left .\Ind{X(T) \geq v}\right |\mathcal{F}_s\right ]\, ,
\]
hence  $P$ is in fact an adapted martingale, and we  obtain the new state variable
\begin{equation}\label{EQN:PCon}
P(s) = q + \int_t^T \alpha_P(s) \,\mathrm{d}W(s)\, ,
\end{equation}
where $\alpha_P$, taking values in $\RR$, is a new control which, {\it a priori}, cannot be assumed  to be bounded, being derived from the martingale representation theorem.
\begin{Remark}
Since $P$ represents the probability required to satisfy a terminal constraint, we could have defined $P$ in equation \eqref{EQN:PCon} as
\[
P(s) = q + \int_t^T P(s)\left (1-P(s)\right )\alpha_P(s)\, \mathrm{d}W(s)\, ,
\]
so that $P$ lies in $[0,1]$.
\end{Remark}

Before explicitly derive the HJB equation we are interested in and following \cite{Bou,Bou2,Son}, let us
further simplify our setting by introducing the set
\[
\begin{split}
D &= \{(t,x,q) \in [0,T]\times \RR^n \times [0,1] \quad : \quad\Ind{X^i(T) \geq v^i}-P^i(T) \geq 0 \quad \mathbb{P}\, \mbox{a.s.} \}\, ,
\end{split}
\]
along with considering the new state variable $P$, see equation \eqref{EQN:PCon}, in such a way that, via the \textit{geometric dynamic programming principle} proved in \cite{Son}, we can define the value function
\begin{equation}\label{EQN:VFunc}
\begin{split}
V(t,x,q) &= \inf \bigg \{ \frac{1}{2} \mathbb{E}_t \left [\sum_{i=1}^n\int_{t}^{T} \alpha^i(s)^2 \,\mathrm{d}s \right ]\quad \mbox{s.t.}\quad \Ind{X^i(T) \geq v^i} -P^i(T) \geq 0 \quad \mathbb{P}\, \mbox{a.s.} \bigg\}\, ,
\end{split}
\end{equation}
where $\mathbb{E}_t$ is the conditional expectation w.r.t. the filtration $\mathcal{F}_t$.\\
Since $V$ is non-decreasing in $q$, we have 
\[
V(t,x,0)\leq V(t,x,q) \leq V(t,x,1)\, , \quad q \in (0,1)\, ,
\]
therefore $V(t,x,0)$ corresponds to the unconstrained problem and its value function is given by $V(t,x,0) = 0$. As regards to the upper bound, we set $V(t,x,1)= \infty$, and 
we prolong the value function outside $[0,1]$, setting $V(t,x,q) = 0$, resp. $V(t,x,q)=\infty$, for $q <0$, resp. for $q>1$.\\
Let us then introduce the Hamiltonian that must be satisfied by the unconstrained optimal control
\begin{equation}\label{EQN:CHam}
H^X(x,\alpha,p,Q_x)=(\mu x + \alpha)\cdot \,p+ \frac{1}{2}\sigma^2 \,x^2\, Q_x+ \frac{1}{2}\|\alpha\|^2\,,
\end{equation}
where above we have denoted for short 
\[
\mu x := (\mu^1 x^1,\dots,\mu^n x^n)\,,
\]
and 
\[
\sigma^2 \,x^2 = diag((\sigma^1 x^1)^2,\dots,(\sigma^n x^n)^2)\,,
\]
being $diag$ the $n \times n$ diagonal matrix.\\
Intuitively, we are expecting that, when the terminal constraint is satisfied, one can solve the classical associated HJB equation whose Hamiltonian is given in equation \eqref{EQN:CHam}, deriving that the optimal control coincides with the unconstrained case. Notice that the optimal solution to the present problem is $\alpha = 0$.\\
As regard the constrained case, taking into account the new martingale process $P$, we have to consider the couple
\[
\begin{split}
\mathrm{d}X^i(s) &= \left (\mu^i\, X^i(s)  + \alpha^i(s) \right ) \mathrm{d}s + \sigma^i\, X^i(s)\,\mathrm{d}W^i(s)\, ,\\
\mathrm{d}P^i(s) &= \alpha_P^i(s)\, \mathrm{d}W^i(s)\, ,
\end{split}
\]
so that we can define the constrained Hamiltonian as
\begin{equation}\label{EQN:ConHam}
\begin{split}
&H^{(X,P)}(x,\alpha,p,Q_x,\alpha_P,Q_{xq},Q_q)=\\
&=(\mu x + \alpha) p+ \frac{1}{2}\sigma^2 x^2 Q+ \frac{1}{2}\|\alpha\|^2 + \sigma x Q_{xq}\alpha_P+\frac{1}{2}\alpha^2_P Q_q\,,
\end{split}
\end{equation}
which should play the role of the Hamiltonian of the associated problem when the constraint is binding. Therefore, the HJB associated to the optimal control reads as follow
\begin{equation}
-\partial_t V-\inf_{\alpha \in \mathcal{A}}\inf_{\alpha_P \in \RR} H^{(X,P)}(x,\alpha,\partial_xV,\partial^2_x V,\alpha_P,\partial^2_{xq} V,\partial^2_q V)=0\,,
\end{equation}
where, above and in what follows, for the ease of notation we avoided writing explicitly the dependencies of $V(t,x,q)$.\\
As mentioned above, $\alpha_P$ 
could be unbounded, implying that
the associated Hamiltonian may be infinite. 
Since the following holds 
\[
H^{(X,P)}(x,\alpha,p,Q_x,\alpha_P,Q_{xq},Q_q) \geq H^X(x,\alpha,p,Q_x)\,,
\]
to evaluate the minimum of $H^{(X,P)}$ w.r.t. $\alpha_P$, we can exploit a first order optimality condition that
\[
\alpha_P = - \sigma\, x\,\frac{Q_{xq}}{Q_q}\, ,
\]
which, when plugged 
into equation \eqref{EQN:ConHam}, gives the following minimum for $H^{(X,P)}$ 
{\footnotesize
\begin{equation}\label{EQN:HBar}
\begin{split}
&\inf_{\alpha_P \in \RR} H^{(X,P)}= \bar{H}(x,\alpha,p,Q_x,Q_{xq},Q_q)= \\
&=
\begin{cases}
(\mu x + \alpha) p+ \frac{1}{2}\sigma^2 x^2 Q_x+ \frac{1}{2}\|\alpha\|^2 - \frac{1}{2Q_q} \sigma^2x^2Q^2_{xq} & Q_q >0\,,\\
(\mu x + \alpha) p+ \frac{1}{2}\sigma^2 x^2 Q_x+ \frac{1}{2}\|\alpha\|^2 & Q_q = 0 \, ,\\
-\infty & \mbox{otherwise}\, .
\end{cases}
\end{split}
\end{equation}
}
It follows that the associated value function introduced in eq. \eqref{EQN:VFunc} solves the following HJB equation
\begin{equation}\label{EQN:HJB1}
-\partial_t V-\inf_{\alpha \in \mathcal{A}} \bar{H}(x,\alpha,\partial_x V,\partial^2_x V,\partial^2_{xq} V,\partial^2_q V)=0\,,
\end{equation}
subject to the terminal condition
\[
V(T,x,q)= 
\begin{cases}
0 & x \geq v\,,\\
\infty & \mbox{otherwise}\,,\\
\end{cases}
\]
where the Hamiltonian $\bar{H}$ is defined as in equation \eqref{EQN:HBar}.

\subsection{The affine control case}

In order to obtain a closed form solution for the HJB equation \eqref{EQN:HBar} we will  further assume that the admissible controls are of the form
\begin{equation}\label{EQN:AffCase}
\alpha^i(t) = \psi^i(t) X^i(t)\,,
\end{equation}
for a fixed constant $\psi^i \in [0,\Psi]$, $\Psi \in \RR_+ \cup \{\infty\}$. From a financial point of view this implies that the LOLR can decide the interest rate at which the banks assets accrues, allowing the bank to have a higher interest rate  to lower the probability of failure.\\
In what follows we derive the explicit solution for the optimal rate $\psi$  the LOLR has to give to the each bank in order to guarantee its terminal survival probability.\\
The strategy is as follows: given the structure of the optimal control problem, we can analyse each node $i$ separately, where we ansatz the value function to be of the form $V(t,x,q)=\sum_{i=1}^n V^i(t,x^i,q^i)$, where each $V^i$ is regarded as the value function for the optimal problem with respect to the element $i$. Thus, for each player $i$ we compute the solution to the above problem in terms of contour line of a function $\gamma^i(t,x,q)$, defining first the boundaries of the domain for the value function $V^i$, then computing explicitly the contour line on the interior of the domain. We stress that in what follows computations are performed independently for any bank $i$, nonetheless for the sake of brevity we will omit the index $i$.\\
Notice first that given an initial datum and a required survival probability $\bar{q}$, it holds that $\mathbb{P}(X(T) \geq v(T)) \geq \bar{q}$, then the optimal control is given by $\psi \equiv 0$, then the value function  $V(t,x,q) \equiv 0$. Therefore we compute three different domains, obtaining, in closed form, two switching curves splitting such domains. The first region $\Gamma_0$ is the region in which the constraint is not binding, implying that the optimal control is given by $\psi \equiv 0$. Financially speaking, whenever the value of the bank lies within the region $\Gamma_0$,  the bank satisfies the LOLR requirement regarding survival probability meaning that it does not need  further help to increase its liquidity. Recall that, the more the value of the bank increases, the safer is the bank.\\
The second region 
is characterized by the condition $\Gamma_\Psi$. In this region the optimal control exceed the maximum rate $\Psi$ the LOLR is willing to grant, implying that the terminal constraint is not satisfied and the value function $V$ diverges. The last domain, denoted by $\Gamma$, is characterized by a binding  terminal constraint, and here the optimal control $\psi \in (0,\Psi)$ has to be explicitly computed. Similarly, we will denote by $\gamma_0$, resp. $\gamma_\Psi$, the switching region between $\Gamma_0$ and $\Gamma$, resp. between $\Gamma$ and $\Gamma_\Psi$.\\
Regarding $\Gamma$ let us define the {\it highest reachable probability} for node $i$ as
\[
\begin{split}
W^H (t,x) &:= \sup \left \{ q \, : V(t,x,q)<\infty \right \} =\sup_{\psi \in [0,\Psi]} \mathbb{P}\left (X^{t,x;\psi}(T) \geq v(T)\right ) \,,
\end{split}
\]
where $X^{t,x;\psi}(T)$  denotes the value at time $T$ with initial datum $(t,x)$ and control $\psi\in [0,\Psi]$. It follows that the highest reachable probability is attained when considering the maximum admissible control $\Psi<\infty$, so that by It\^{o} formula and the \textit{Feynman--Kac theorem}, we have that $W^H (t,x)$ solves the parabolic PDE
\[
\begin{cases}
W^H (t,x)(T,x) &= \Ind{[v(T),\infty)}(x)\,,\\
-\partial_t W^H (t,x) &= \partial_x W^H (t,x) \left (\mu + \Psi\right )x + \frac{1}{2} \sigma^2 x^2 \partial_x^2 W^H (t,x)\,,
\end{cases}
\]
whose solution can be explicitly computed to be

\begin{equation}\label{EQN:Wpsiq}
\begin{split}
&W^H (t,x) = \mathbb{P}\left (\log X^{t,x;\Psi}(T) \geq \log v(T) \right ) = \\
&=\mathbb{P}\left (W(T-t) \geq \frac{1}{\sigma}\left (\log \frac{v(T)}{x} - \left (\mu + \Psi-\frac{\sigma^2}{2}\right )(T-t)\right )\right ) =\\
&= \frac{1}{2} \left (1- \text{Erf} \left ( \frac{\log \frac{v(T)}{x} - \left (\mu + \Psi-\frac{\sigma^2}{2}\right )(T-t)}{\sqrt{2\sigma^2(T-t)}}\right )\right ) =\\
&= \frac{1}{2} \left (1- \text{Erf}\left (d(\mu,\Psi,\sigma,T-t\right )\right )\,,
\end{split}
\end{equation}

with
\[
d(\mu,\Psi,\sigma,T-t) :=  \frac{\log \frac{v(T)}{x} - \left (\mu + \Psi-\frac{\sigma^2}{2}\right )(T-t)}{\sqrt{2 \sigma^2(T-t)}}\,,
\]
and $Erf$ denotes the \textit{error function}.
For $W^H(t,x) = \bar{q} \in (0,1)$, we have that
\[
\frac{1}{2} \left (1- \text{Erf}\left (d(\mu,\Psi,\sigma,T-t\right )\right ) = \bar{q}\,,
\]
and solving for $\Psi$, we obtain the boundary region in implicit form
\begin{equation}\label{EQN:Psitx1}
\Psi=\gamma_\Psi(t,x;\bar{q}) = \left (\frac{\sigma^2}{2} - \mu \right ) + \frac{\log \frac{v(T)}{x}}{T-t} - \frac{\sigma \rho}{\sqrt{T-t}}\,, 
\end{equation}
with
\[
\rho := \sqrt{2}\,\text{Erf}^{-1}\left (1-2\bar{q}\right )\,.
\]
Thus, for a required probability of success $\bar{q}$, the control problem is not feasible in
\[
\Gamma_\Psi = \left \{ (t,x) \, : \left (\frac{\sigma^2}{2} - \mu \right ) + \frac{\log \frac{v(T)}{x}}{T-t} - \frac{\sigma \rho}{\sqrt{T-t}} > \Psi \right \} \, 
\]  
so that, for starting data $(t,x)$ within the left hand side of $\gamma_\Psi(t,x;\bar{q})$, see eq. \eqref{EQN:Psitx1}, the terminal constraint cannot be satisfied, see Figure \ref{FIG:Region}. If $\Psi = \infty$, that is the LOLR is willing to give a possibly infinite return rate, any point is controllable, hence we can always find an admissible control such that the terminal probability constraint is attained.\\
As regard $\Gamma_0$, computing the no-action region we have
\[
\begin{split}
W^0(t,x) &= \mathbb{P}\left (X^{t,x;\psi_0}(T) \geq v(T)\right ) =\\
&= \frac{1}{2} \left (1- \text{Erf}\left (d(\mu,\psi_0,\sigma,T-t\right )\right )\,,
\end{split}
\]
then
by assuming that $W^0(t,x) = \bar{q} \in (0,1)$, we have 
\[
\frac{1}{2} \left (1- \text{Erf}\left (d(\mu,\psi_0,\sigma,T-t\right )\right ) = \bar{q}\,,
\]
and, solving for $\psi_0$, we obtain the boundary region
\begin{equation}\label{EQN:Psitx}
0=\gamma_0(t,x;\bar{q})= \left (\frac{\sigma^2}{2} - \mu \right ) + \frac{\log \frac{v(T)}{x}}{T-t} - \frac{\sigma \rho}{\sqrt{T-t}}\,, 
\end{equation}
where
\[
\rho := \sqrt{2}\,\text{Erf}^{-1}\left (1-2\bar{q}\right )\,,
\]
and, as done before, we are left with the following
no-action region
\[
\Gamma_0 = \left \{ (t,x) \, : \left (\frac{\sigma^2}{2} - \mu \right ) + \frac{\log \frac{v(T)}{x}}{T-t} - \frac{\sigma \rho}{\sqrt{T-t}} < 0 \right \} \, 
\]
so that, given a starting value $(t,x) \in \Gamma_0$, the terminal constraint is satisfied and the optimal return is given by the null control $\psi \equiv 0$.\\
At last the action region $\Gamma$ is the one delimited by $\Gamma_0$ and $\Gamma_\Psi$, that is
\[
\Gamma = \left \{ (t,x) \, : 0 < \left (\frac{\sigma^2}{2} - \mu \right ) + \frac{\log \frac{v(T)}{x}}{T-t} - \frac{\sigma \rho}{\sqrt{T-t}} < \Psi \right \} \, .
\]
Thus, being $(t,x) \in \Gamma$, the controller has to find the optimal control so that the terminal probability constraint holds. 
By computing the reachability set with  fixed constant control $\bar{\psi}$, that is
\[
\begin{split}
W^{\bar{\psi}}(t,x) &= \mathbb{P}\left (X^{t,x;\bar{\psi}}(T) \geq v(T)\right ) = \mathbb{E}\left [\Ind{[v(T),\infty)}\left (X^{t,x;\bar{\psi}}(T)\right ) \right ]\,,
\end{split}
\]
proceeding as above, we obtain
\begin{equation}\label{EQN:Wpsiq}
\begin{split}
&W^{\bar{\psi}}(t,x) = \mathbb{P}\left (\log X^{t,x;\bar{\psi}}(T) \geq \log v(T) \right ) = \frac{1}{2} \left (1- \text{Erf}\left (d(\mu,\bar{\psi},\sigma,T-t\right )\right )\,,
\end{split}
\end{equation}
which implies
\begin{equation}\label{EQN:BarPsi}
\bar{\psi}=\gamma_{\bar{\psi}}(t,x;\bar{q})= \left (\frac{\sigma^2}{2} - \mu \right ) + \frac{\log \frac{v(T)}{x}}{T-t} - \frac{\sigma \rho}{\sqrt{T-t}}\;. 
\end{equation}
What we have obtained so far
has to be intended as follows: if the autonomous process $X^{t,x;0}(T)$ already satisfies the terminal probability constraint, then it is optimal to solve the control problem with no terminal constraint, whose solution is given by the null control in the present case. \\
Therefore, for a fixed $q \in (0,1)$, if $(t,x) \in \Gamma$,  the optimal control $\psi$ is given by
\begin{equation}\label{EQN:Psitx}
\gamma_\psi(t,x;q) = \left (\frac{\sigma^2}{2} - \mu \right ) + \frac{\log \frac{v(T)}{x}}{T-t} - \frac{\sigma \rho}{\sqrt{T-t}}=\psi\,,
\end{equation}
see Figure \ref{FIG:Region} for a representation of the above obtained regions.
\begin{figure}
\centering
\includegraphics[width=.7\textwidth]{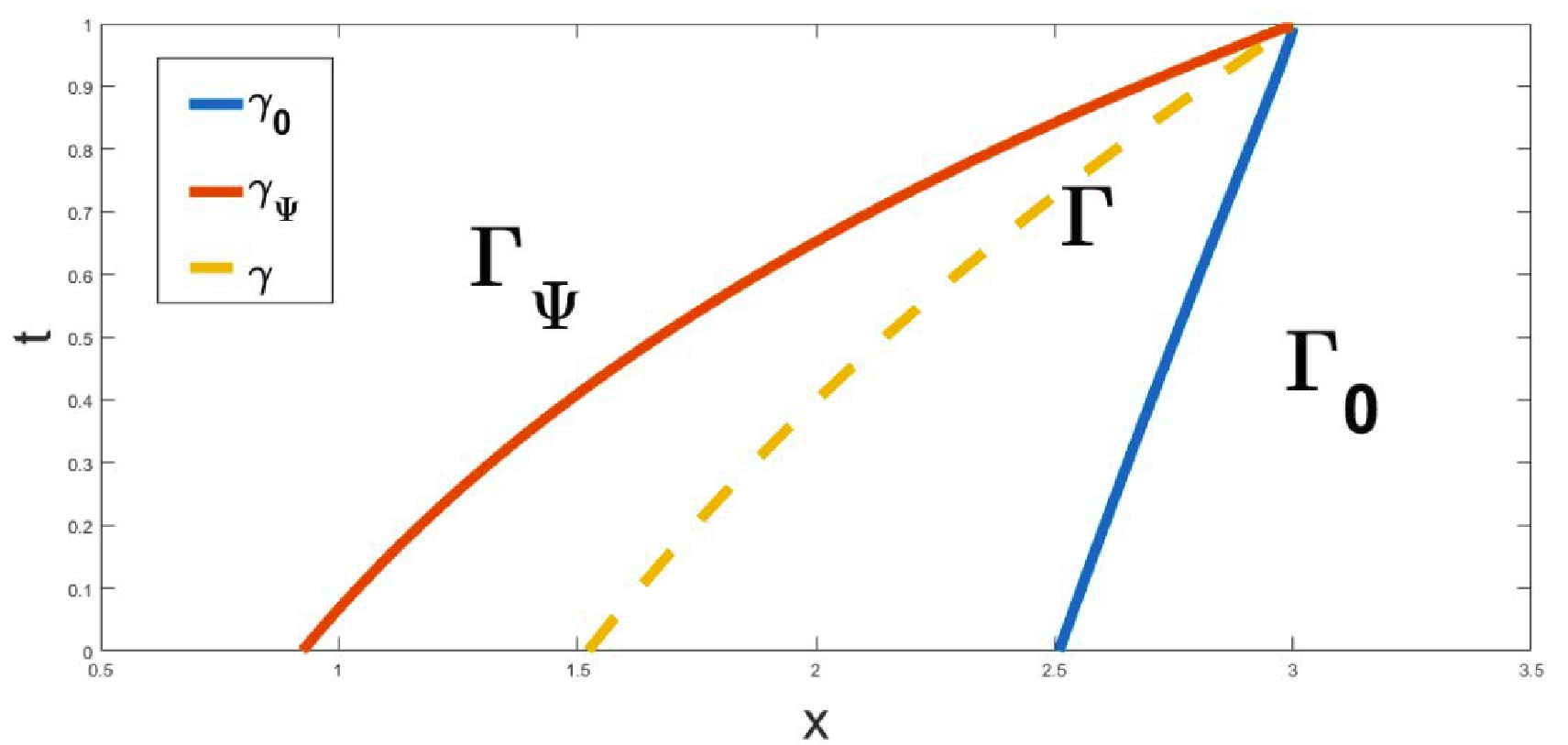}
\caption{Representation of different domains for the optimal control problem.}\label{FIG:Region}
\end{figure}
Moreover, along the curve $W^{\bar{\psi}}(t,x)$, the terminal probability of success remains constant, so that the optimal control is given by the constant control $\bar{\psi}$.\\
Being the optimal control for node $i$ constant along $W^{\bar{\psi}}(t,x)$, then, exploiting equation \eqref{EQN:VFunc}, the value function for the above control problem reads as follow
{\footnotesize
\begin{equation}\label{EQN:SolV}
V^i\left (t,x^i,W^{\bar{\psi}^i}(t,x^i)\right ) = (\bar{\psi}^i)^2 (x^i)^2 \left (\frac{e^{(2(\mu^i + \bar{\psi}^i) + (\sigma^i)^2) (t-T)} -1}{2(\mu^i + \bar{\psi}^i) + (\sigma^i)^2 }\right ) \,,
\end{equation}
}
so that we have the following.
\begin{Theorem}
The value function for the optimal control problem \eqref{EQN:CrirMin} is given by
\begin{equation}\label{EQN:SolVFin}
V(t,x,W^{\bar{\psi}}(t,x))=\sum_{i=1}^n V^i(t,x^i,W^{\bar{\psi}^i}(t,x^i))\,,
\end{equation}
where
\begin{description}
\item[(i)] if $(t,x^i) \in \Gamma^i$ and $q^i \in (0,1)$ are such that $\gamma_{\bar{\psi}^i}(t,x^i,q) = \bar{\psi}^i$, then $\bar{V}(t,x^i,W^{\bar{\psi}^i}(t,x^i))$ is given as in equation \eqref{EQN:SolV}.
\item[(ii)] if $(t,x^i) \in \Gamma_0^i$ and $q^i \in (0,1)$, then it holds $V^i(t,x^i,q^i)=0$.
\end{description}
Then $V$ as defined above in equation \eqref{EQN:SolVFin} defines a classical solution to the HJB equation \eqref{EQN:HJB1} on $\Gamma \cap \Gamma_0$.\\
Moreover, the optimal control within the class of affine controls is given as in equation \eqref{EQN:AffCase}, where $\psi$ is given as in equation \eqref{EQN:BarPsi}
\end{Theorem}
\begin{proof}
The structure of the optimal control problem gives that the contribute of each node can be treated separately, so that the value function is of the form \eqref{EQN:SolVFin}, where each $V^i$ can be regarded as the value function for the optimal control for the node $i$ alone. As above, for ease of notation, we will omit the index i.\\
Fixing the node $i$, it can be trivially shown that for $(t,x) \in \Gamma_0$ we have $V(t,x,q)=0$.\\
Let $(t,x) \in \Gamma$, thus along $W^{\bar{\psi}}(t,x)$, the terminal probability of surviving is fixed, so that explicit computation shows that $V$ as defined in equation \eqref{EQN:SolV} solves the HJB equation \eqref{EQN:HJB1}. Observing that the map
\[
q \mapsto V(t,x,q)\,,
\]
is non--decreasing, together with the fact that $W^{\bar{\psi}}(t,x) > W^{\psi}(t,x) $ for $\bar{\psi} > \psi$, we have that 
\[
V\left (t,x,W^{\psi}(t,x) \right ) = -\infty\,, \quad \psi < \bar{\psi}\,,
\]
because the terminal constraint in equation \eqref{EQN:VFunc} is not satisfied. Analogously, if $\psi > \bar{\psi}$, then 
$W^{\bar{\psi}}(t,x) < W^{\psi}(t,x)$. Therefore, as before, the non-decreasing property of $V$ w.r.t. the third argument $q$, implies 
\[
V\left (t,x,W^{\psi}(t,x) \right ) > V\left (t,x,W^{\bar{\psi}}(t,x) \right )\,,
\]
and the minimum is attained for the control $\bar{\psi}$ implicitly given by equation \eqref{EQN:Psitx}.\\
As regard the value function regularity, notice that it is a classical solution in both region $\Gamma$ and $\Gamma_0$. In order to prove that it is a global classical solution we need to prove that it is regular on $\gamma_0$. Let $\bar{x}$ the value on the switching curve $\gamma_0$, that is for fixed $(t,q)$, we have that $\gamma_0(t,\bar{x},q)=0$; then since $\bar{\psi} \to 0$ as $x \to \bar{x}^-$ we have that $\lim_{x \to \bar{x}^-} \partial_x^2 V = 0 = \lim_{x \to \bar{x}^+} \partial_x^2 V$ and $\lim_{x \to \bar{x}^-} \partial_x V = 0 = \lim_{x \to \bar{x}^+} \partial_x V$, hence the value function is differentiable on $\Gamma \cup \Gamma_0$.
\end{proof}
\section{Application to a network of financial banks}\label{Appl}
In the present section we use previously obtained results to study a real-world application
characterized by an interconnected network of banks. In particular, we will show how optimal solutions previously computed can modify the evolution of sucha a network.
We stress that, for the sake of readability, we will apply our results to a small network, even if, due the fact that the optimal solution is computed in closed form, our results can be easily extended to arbitrary {\it big} systems.\\
\subsection{PageRank}\label{SEC:PR}
Before introducing the model, let us introduce an explicit method to address relative importance of a single node in a network. In particular, such an approach will be then used to systematically decide the survival probability for each node.\\
Let us note that, along  previous sections, we have stated an optimal control problem which has been then solved  deriving its solution under the assumption that the \textit{accepted probability of failure} $q^i$ is a fixed parameter to be chosen endogenously. In what follows we propose a general, automatic, criterion to deduce the global importance of each node in the system. Next computations exploit results on network analysis already  used, e.g., to set the functioning logic of the \textit{Google} research engine, see, e.g., \cite{Page}. According to the network formulation introduced in Section \ref{SEC:GS}, and using results derived in \cite{Page}, we show how to score the relative importance of any bank in the network, computing its so called \textit{Page Rank}, allowing us to choose 
the best survival probability $q$.\\
According to the framework described in Section \ref{SEC:GS}, see also Appendix \ref{G:framework},
let us consider a system of interconnected $n$ banks and related
standard  {\it bank enumeration}. Namely, we take into account the usual {\it one-to-one} correspondence relation between the set of banks and the set of vertexes $V:=\{v_1,v_2,\dots,v_n\}$,  referred to as nodes, while  $I:=\{1,2,\dots,n\}$ is the associated set of indexes. Moreover, consider a LOLR strategy in which for each $v_i\in V$ the default probability constraint parameter $q^i$ depends on a predetermined rank $R^i$ associated to the $i^\text{th}$ bank, hence representing its systemic importance in the network.\\
In what follows we are considering graphs as  defined in Section \ref{SEC:GS}. In particular, to each node $v_i\in V$ corresponds a bank, while to edges connecting nodes $(v_i,v_j)\in V\times V$, we associate the following quantities 
\begin{equation}\label{gamma}
\gamma_{(i,j)}^+=\frac{c^+\,L_{i,j}+c^-\,L_{j,i}}{N_j-\min(N)+1},\quad\gamma_{(i,j)}^-=\frac{c^+\,L_{j,i}+c^-\,L_{i,j}}{N_i-\min(N)+1},
\end{equation}
where, letting
\begin{align}
L_j^+=\sum_{i\sim j}L_{ij},\quad L_j^-=\sum_{i\sim j}L_{ji},\\
i\sim j \iff v_i,v_j \text{ are connected},\nonumber
\end{align}
we define $N_j$ as the net amount of money held by bank $j$ if it would pay its debts at the actual time, i.e. $N_j:=X_j+L_j^+-L_j^-$.
Moreover $c^+$ and $c^-$ are two non-negative constants chosen to confer more importance to due debts, resp. to owed credits. For the sake of simplicity, since $c^+$ and $c^-$ are meant to be weight parameters, we set $c^+ + c^-=1$. Notice that $\gamma_{(i,i)}^+=\gamma_{(i,i)}^-=0$ and $\gamma_{(i,j)}^-=\gamma_{(j,i)}^+$, for all $i,j\in I$.\\
Let us  introduce the notion of {\it outdegree} $\text{deg}_\gamma^+$, resp. {\it indegree} $\text{deg}_\gamma^-$, for any vertex $v_i\in V$, namely
\[
\text{deg}_\gamma^+(v_i)=\sum_{j\in\mathcal{I}}\gamma_{(i,j)}^+,\quad\text{deg}_\gamma^-(v_i)=\sum_{j\in\mathcal{I}}\gamma_{(i,j)}^-\,,
\]
and normalize the quantities defined in \eqref{gamma} associated to any couple $(i,j)$ of edges in the graph
\[
\vec{\tau}_{(i,j)}=\frac{\gamma_{(i,j)}^+}{\text{deg}_\gamma^+(v_j)},\quad \cev{\tau}_{(i,j)}=\frac{\gamma_{(i,j)}^-}{\text{deg}_\gamma^-(v_j)}
\]
corresponding to the ratio of a linear combination on the liabilities between bank $i$ and bank $j$, and the asset value of bank $j$. Moreover, we define the matrix $\vec{\mathcal{T}}$ as the matrix whose entries are $\vec{\tau}_{(i,j)}$, for $i,j\in\mathcal{I}$, the quantities $\vec{\tau}_{(i,j)}$ being  the weights assigned to each oriented edge. \\
Therefore, the rating value associated to any node/bank $v_i$ is given by the following recursive formula 
\begin{equation}\label{rating}
R^i_d=d\,\sum_{j\sim i}\vec{\tau}_{(i,j)}\,R^j_d,
\end{equation}
where $d\in(0,1)$ is a parameter to be chosen,  typically  $d=0.85$, see, e.g., \cite{Mug}. To compute equation \eqref{rating}, we  introduce the so called \emph{Google-matrix}, see, e.g., \cite[ch 2]{Mug}.\\
We assume that our network is composed by banks not solely owing liabilities if $c^+=1$, resp. not solely owning liabilities if $c^+=0$, and at least connected for $c^+\in(0,1)$.
Of course,  banks that are non connected  to others belonging to the network, are simply not ranked, since their default cannot affect the system. On the other hand, even if the conditions for $c^+\in\{0,1\}$ are not required, they guarantee the boundedness of all the elements of the matrix defined in the next Definition \ref{D:Gm}. We stress that,  to avoid above restrictions, one can modify the values  assigned  to edges by equation \eqref{gamma}, e.g., as follows: for $c^+=1$ and for every $i\sim j$, define  $\widetilde{\gamma}_{(i,j)}^+=L_{i,j}/(N_j-\min(N)+1)+\epsilon$  as the modified value assigned to the edges.
\begin{Definition}[Google-matrix]\label{D:Gm}
Let $J$ be a $n\times n$-matrix whose entries are all ones. A \textit{Google-matrix} is a $n\times n$-matrix given by
\begin{equation}\label{Gm}
\mathcal{G}_d:=\frac{1-d}{n}\,J+d\,\vec{\mathcal{T}},
\end{equation}
where $d\in(0,1)$ can be chosen  to guarantee irreducibility of $\mathcal{G}_d$, while $J$ is the $n \times n$ matrix whose all entry are $1$. 
\end{Definition}
Since the matrix defined in equation \eqref{Gm} is positive we can apply the Perron–Frobenius Theorem which assures us that there exists a {\it maximum} real eigenvalue $\lambda>0$ of $\mathcal{G}_d$,
indeed $\lambda$ is the so-called dominant Perron–Frobenius eigenvalue. Moreover, there exists one of the associated eigenvectors, denoted by $R_d$ and usually called \emph{Perron-Frobenius dominant vector}, which is both strictly positive and normalized and whose components represent the rating of each bank. Let us recall that $d$ is usually chosen to be approximately equals to 0.85, see, e.g., \cite{Mug}.\\
It follows that proposed ranking procedure  consists in computing the following series
\[
R_d=d\,\sum_{k=0}^\infty (1-d)^k\,(\mathcal{G}_d)^k\,\mathbf{1}\,,
\]
where we denoted by $\mathbf{1}$ a $n$-dimensional vector whose entries are all equal to one.

\subsection{A concrete case study}\label{S:PRexample}
In what follows we  consider a systems of banks aiming at computing their ranking. First of all, we are considering a LOLR willing to save banks whose failure would cause insolvency and no ability to pay back their liabilities, i.e. $c^+=0$, implying $c^-=1$. According to what we have  seen along previous sections, see also \cite{Mug}, we fix $d=0.85$ and  we consider a  system of banks whose liability matrix and cash vector are as follows, see Figure \ref{fig} for the associated graph:
\[
\mathcal{L} =
\begin{bmatrix}
0  &0  &10&0  \\
5  &0  &5  &5  \\
0  &0  &0  &0  \\
10&4 &0  &0
\end{bmatrix},\quad 
X=
\begin{bmatrix}
5.2\\6\\13\\3
\end{bmatrix}\;.
\]
Let us note that proxies for $\mathcal{L}$ and $X$ can be evaluated through the methodology proposed in \cite{DocWang} starting by synthetic data generated by FX markets settlements. However, the aim of this paper is to focus on the LOLR strategy which should have complete information on the financial market,  therefore we do not go into technical details on the estimation procedure.\\
As explained in Definition \ref{D:Gm},  the associated Google-matrix can be computed as follows
\[
\mathcal{G}_d=
\begin{bmatrix}
0.0375&0.8344&0.0375&2.3042\\
0.0375&0.0375&0.0375&0.9442\\
2.9352&0.8344&0.0375&0.0375\\
0.0375&0.8344&0.0375&0.0375
\end{bmatrix},
\]
where the eigenvalues of the matrix $\mathcal{G}_d$ are $\lambda_1=1.2892$, $\lambda_2=-0.8449$, $\lambda_3=-0.1472 + 0.3982i$ and $\lambda_4=\overline{\lambda_3}$. The absolute value of the eigenvector corresponding to the highest eigenvalue is
\[
R=v_1=\begin{bmatrix}0.3516&0.1342&0.9177&0.1275\end{bmatrix}^T\;.
\]
The third bank is the one with the highest ranking. Indeed, it is easy to note that its default would cause the default of the first bank and then an insolvency cascade. This is due to the fact that the third bank is systematically more important than the others. Notice that the amount of money due is the most important aspect to be taken into account for the safety of the system. We have reported in  Table \ref{tab1}  further considerations.
\begin{figure}
\centering
\includegraphics[width=.4\textwidth]{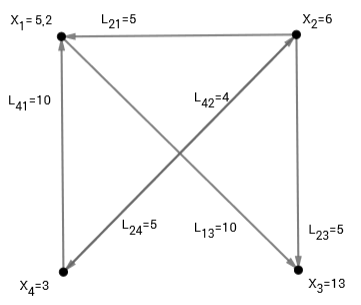}
\caption{Graph representing the system of banks: nodes report the cash value of each bank, while the oriented edges represent the amount of money lend from a bank to another.}\label{fig}
\end{figure}

\begin{table}
\begin{center}
\begin{tabular}{c | c c c c}
Banks  ($i$)&1&2&3&4\\
\hline
$X_i$& 5.2&6&13&3\\
$\sum_{j\sim i}L_{ji}$ &15&4&15&5\\
$R_i$&0.3516&0.1342&0.9177&0.1275
\end{tabular}\caption{ Comparison among the banks rankings.}\label{tab1}
\end{center}
\end{table}
\begin{Remark}
Looking at  Figure \ref{fig} and Table \ref{tab1}, we can see that although the first and the third bank are owning the same amount of money to other banks, nonetheless their rankings $R$ are significantly different. This is due to the fact that Bank 3 owns to Bank 1 and its insolvency would probably cause the default of Bank 1. In this example the cascade effect caused by the default of Bank 3 would stop with the default of two banks because of the small dimension of the system, while, on the contrary, such an effect amplifies in big networks.
\end{Remark}
\subsection{LOLR strategy under the PageRank approach}
\begin{figure}
\centering
\includegraphics[width=.4\textwidth]{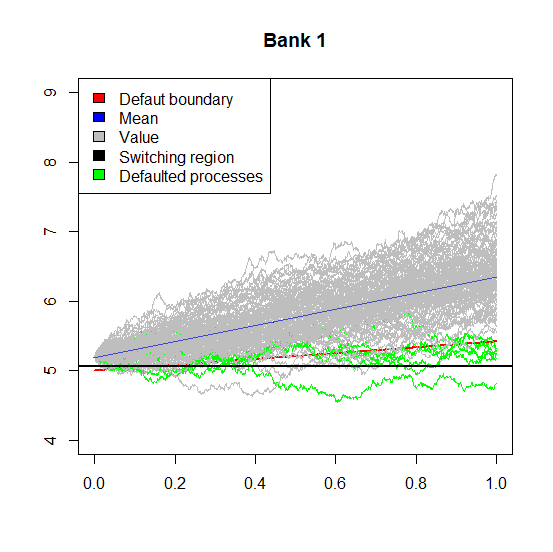}\\
\includegraphics[width=.4\textwidth]{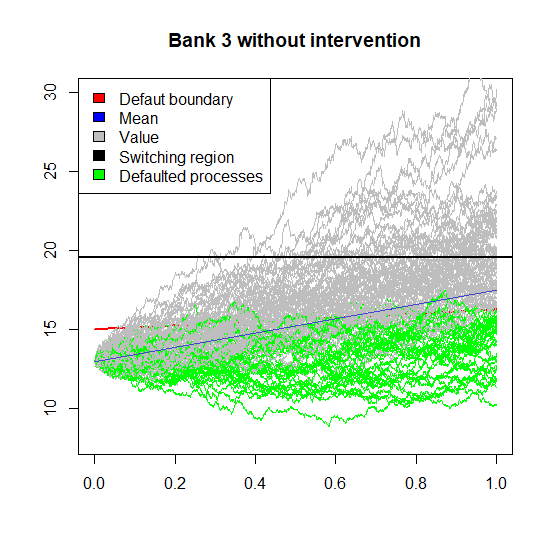}
\includegraphics[width=.4\textwidth]{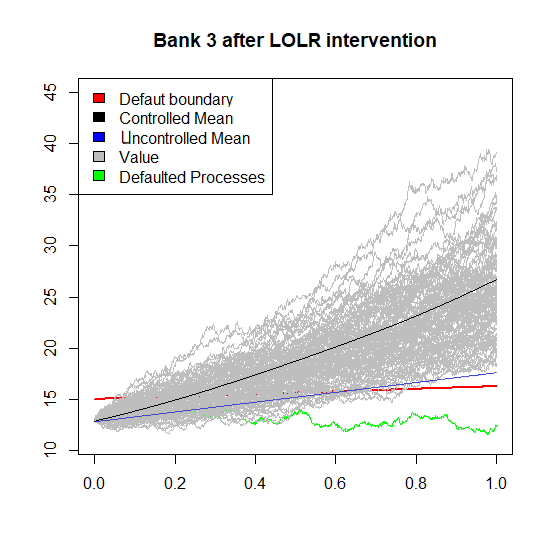}
\caption{100 simulations for the evolution of Bank 1 (top panel), 100 simulations for the evolution of Bank 3, without LOLR intervention (middle panel) and 100 simulations for the evolution of Bank 3, with LOLR intervention (bottom panel).}\label{figBC}
\end{figure}
Iin what follows
we shall describe how to adapt the LOLR problem stated in Section \ref{SEC:PC} to guarantee more flexibility to those banks that are more important for the network's health. Such type of strategies are often referred to as \textit{Systemic importance driven} (SID) strategies, see appendix \ref{C:cap} for more details.\\
We recall that the aim of the LOLR is to minimize the expenditure on banks bailout given by equation \eqref{EQN:CostFunc} constrained by \eqref{EQN:ProbC}, i.e. guaranteeing a probability $q^i$ that the bank $i$ will not default. Let us fix an identical probability constraint $q\in[0,1)$ for all the banks, hence adopting an equality policy analogous to the \textit{max liquidity} (ML) strategy introduced in \cite{Cap}, see also appendix \ref{C:cap}. We note that a ML strategy guarantees no privileges to any banks, which  would lead the LOLR to lend the same amount of money for systematically important banks as for those banks whose failure would not cause a {\it cascading effect}.\\
The main idea of the subsequent analysis consists in defining the probability constraints as an increasing function of the rank assigned to each bank. Namely, we have
\[
q^i=f(R^i),\quad\text{for $f:\mathbb{R}^+\rightarrow[0,1)$ increasing function}\,,
\]
where, as seen in Section \ref{SEC:PR}, $R^i$ is the ranking of the bank $i$. Notice that requiring $f'=0$ the LOLR will again be restricted to the ML strategy.\\
In \cite{Cap} we have shown  that choosing $f$ to be an increasing function leads to a more convenient scenario for the health of networks which have a core-periphery structure, whereas, normally, banks networks have a dense cohesive core, with a periphery less connected.\\
Coming back to the type of network already defined in subsection \ref{S:PRexample}, see Figure \ref{fig}, we assume that the LOLR assigns the following probability constraints
\begin{equation}\label{qR}
q^i=0.9+0.05\cdot\mathbbm{1}_{\{R^i>0.5\}}+0.04\cdot\mathbbm{1}_{\{R^i>0.75\}}\,;
\end{equation}
and we perform a one period simulation of the network, see Figure \ref{fig}, taking $t_0=0$ and $T=1$. Let us assume, for the sake of simplicity,  that all the liabilities expire at time $T$, and that they exponentially increase in time with fixed growth rate $r=0.08$, i.e.
\[L(t)=L\,e^{r\,t},\quad \text{for }t\in[0,1]\,.
\]
Furthermore, we assume that cash vectors' dynamic evolve according to geometric Brownian motions evolving, namely:
\begin{align*}
\mathrm{d}X_t^1&=X_t^1\,(0.2\,\mathrm{d}t+0.1\,\mathrm{d}W_t),\\
\mathrm{d}X_t^2&=X_t^2\,(0.15\,\mathrm{d}t+0.25\,\mathrm{d}W_t),\\
\mathrm{d}X_t^3&=X_t^3\,(0.3\,\mathrm{d}t+0.2\,\mathrm{d}W_t),\\
\mathrm{d}X_t^4&=X_t^4\,(0.05\,\mathrm{d}t+0.4\,\mathrm{d}W_t)\;.
\end{align*}
Then, accordingly to equation \eqref{EQN:Psitx}, we have that the banks' log-switching regions $y^i$, $i=1,\ldots,4$, read as follow 
\begin{align*}
y^1=1.622593, \quad\quad y^2=0,\quad y^3=2.97332,\quad\quad y^4=0,\;\\
q^1=0.9,\quad\quad\quad q^2=0.9,\quad\;\; q^3=0.99,\quad\quad q^4=0.9.
\end{align*}
recalling that they have to be less than the log initial wealth $X^i(0)$  in order to guarantee the fulfilment of the probability constraint. Therefore, since 
\[
\begin{split}
\log(X^1(0))&=1.6487,\quad\log(X^2(0))=1.7918\,\\
\log(X^3(0))&=2.5649,\quad\log(X^4(0))=1.0986\,,
\end{split}
\]
we have that the LOLR has to intervene controlling Bank 3. Notice that the LOLR has not to intervene in banks 2 and 4, since they have more credits than debits, hence they cannot face bankruptcy, while the opposite is true for banks 1 and 3. For $q^1=0.95$, we would have $\widetilde{y}^1=1.6589$ and there would need a LOLR intervention injecting money also in bank 1.\\
Figure \ref{figBC} (top panel) represents 100 simulation for the evolution of Banks 1 and 3, with and without LOLR intervention. Since the probability of Bank 1 to survive is greater than $q^1=0.9$, the LOLR is not going to intervene, whereas indeed its probability to default is approximately $0.062$. \\
Clearly, requiring $q^1 =0.95$ would imply that the LOLR has to intervene lending money to Bank 1. In the middle Figure \ref{figBC}, there are represented 100 simulations of the process associated to Bank 3; since $q^3=99\%$ and the default probability of Bank 3 is $0.388$, the LOLR is going to intervene injecting capital in its cash reserve. After the optimal injection of capital, Bank 3 has probability $0.01$ to face the default event, see the lower Figure \ref{figBC}, for the representation of 100 simulations of Bank 3 in the case in which the LOLR is going to intervene. Let us underline that the simple case-study we analysed  has been setted to provide an example as clear as possible, nonetheless, because all the analytical results we derived are in closed form, general complex networks can be theoretically treated as well. Clearly, increasing the graph connection grade, we have an exponential growth in computing the quantities of interest.

\section{Conclusions}
In the present work, we have derived a closed form solution for an optimal control of interbank lending subject to specific  terminal probability constraints on the failure of a bank. The obtained result can be applied to a system of interconnected banks, providing a network solution.\\
We have also shown a simple and direct method to derive the relative importance of any {\it node}  within the studied network. We would like to underline that such a {\it ranking value} is fundamental in deciding the accepted probability of failure which modifies the final optimal strategy of a financial supervisor aiming at controlling the system to prevent {\it global crisis} as generalized default.\\
The results here presented  constitute a first step of a wider research program. In particular, in future works we shall consider 
sequence of {\it checking times} each of which characterized by possibly different constraints to be considered by the supervisor.
In this setting, a solution can be obtained by a backward induction approach, see \cite{CDP,Pha}, applied to results here derived.
Moreover, as a further development we will  consider a framework where  the failure can happen continuously in time, hence imposing  strict constraints at any time before the terminal one $T$.

\appendix
\section{General framework for systemic risk in financial networks}\label{G:framework}
Let us first introduce the mathematical notation needed to properly treat the general financial scenario we are interested in. In particular, we consider a finite connected financial network identified with a graph $\GG$ composed by $n \in \NN$ vertices $\Ver_1,\dots,\Ver_n$, corresponding to $n$ banks, and $m \in \NN$ edges $e_1, \dots, e_m$ assumed to be normalized on the interval $[0,1]$, which represents interaction between the $n$ banks. In what follows we will use the Greeks letters $\alpha, \, \beta, \, \gamma=1,\dots,m$  to denote edges, whereas  $i, \, j, \, k = 1,\dots,n$, will denote vertexes. We refer to \cite{CDP3,CDP4,Mug}, for further details\\
The structure of the graph is based on the  \textit{incidence matrix} $\Phi:= \Phi^+-\Phi^-$, where the sum is intended componentwise and $\Phi = \left (\phi_{i,\alpha}\right )_{n \times m}$, together with the \textit{incoming incidence matrix} $\Phi^+ = \left (\phi_{i,\alpha}^+\right )_{n \times m}$, and the \textit{outgoing incidence matrix} $\Phi^- = \left (\phi_{i,\alpha}^-\right )_{n \times m}$, where
\[
\phi^+_{i,\alpha} = 
\begin{cases}
1 & \Ver_i = e_\alpha(0)\, ,\\
0 & \mbox{ otherwise } 
\end{cases} \, ,\quad 
\phi^-_{i,\alpha} = 
\begin{cases}
1 & \Ver_i = e_\alpha(1)\, ,\\
0 & \mbox{ otherwise }  \, .
\end{cases}
\]
In particular, we will say that the edge $e_\alpha$ is \textit{incident} to the vertex $\Ver_i$ if $|\phi_{i,\alpha}|=1$, so that
\[
\Gamma(\Ver_i) = \{\alpha \in \{1,\dots,m\} \, : \, |\phi_{i,\alpha}| =1\}\, ,
\]
represents the set of incident edges to the vertex $\Ver_i$. We also introduce the \textit{adjacency matrix} $\mathcal{I} = \left (\iota_{i,j}\right )_{n \times n}$, defined as $\mathcal{I} := \mathcal{I}^+ + \mathcal{I}^-$, where $\mathcal{I}^+=\left (\iota^+_{i,j}\right )_{n \times n}$, resp. $\mathcal{I}^- = \left (\iota^-_{i,j}\right )_{n \times n}$, is the \textit{incoming adjacency matrix}, resp. \textit{outgoing adjacency matrix}, defined as
\[
\iota^+_{i,j} = 
\begin{cases}
1 &  \mbox{ it exists } \, \alpha = 1,\dots, m \, : \, \Ver_j = e_\alpha (1) \,, \, \Ver_i = e_\alpha (0) \, ,\\
0 & \mbox{ otherwise }\, ,\\
\end{cases}
\]
\[
\iota^-_{i,j} = 
\begin{cases}
1 &  \mbox{ it exists } \, \alpha = 1,\dots, m \, : \, \Ver_j = e_\alpha (0) \,, \, \Ver_i = e_\alpha (1) \, .\\
0 & \mbox{ otherwise }\, .
\end{cases}
\]
Notice that since $\mathcal{I}^+ = (\mathcal{I}^-)^T$, then we have that $\mathcal{I}$ is symmetric with null entries on the main diagonal.\\

\section{Comparison with the paper by Capponi et al. \cite{Cap}}\label{C:cap}

As mentioned above, the financial setting has been mainly borrowed by \cite{Eis} as concerns the lending system formulation, and from \cite{Cap} for the optimal control problem with an external supervisor aiming at guaranteeing the overall sanity of the system.

This section is devoted to a comparison with \cite{Cap}. We stress that our assumptions on the optimal control are in the spirit of \cite{Cap}, in the sense that we consider failure at discrete times; also we will not consider a global optimal control, deriving a control for the whole time interval but rather we derive a series optimal control and then gluing together the resulting optimal controls. As mentioned we leave the optimal global control to future research being this latter point mathematically more demanding.

This comparison is significant since their work is based on a similar framework, namely a multi-period controlled system of banks, represented by a network, in which an outside entity, named LOLR, provides liquidity assistance loans to financially unstable banks in order to reduce the level of systemic risk within the whole network of banks. 
To analyze the systemic risk in interbank networks their work follows a clearing system framework consistent with bankruptcy laws. In particular they generalize the single period clearing system in the paper by Eisenberg and Thomas, see \cite{Eis}, by a multi-period controlled clearing payment system assuming limited liability of equity, priority over equity, and proportional repayments of liabilities after the default event. This generalization leads to a better insight in the propagation and aftershocks of defaults.
The main feature in \cite{Cap} is the comparison between two possible LOLR strategies:
\begin{itemize}
\item the \emph{Systemic Importance Driven} (SID) strategy, in which liquidity assistance is available only to banks considered systemically important, i.e. the banks whose default would cause significant losses to the financial system (because of their size, complexity and systemic interconnectedness);
\item the \emph{Max-Liquidity} (ML) strategy, in which the regulators aim to maximize the instantaneous total liquidity of the system.
\end{itemize}
By the analysis of these two different strategies they showed that the SID strategy is preferred when the network has a core-periphery structure, i.e. consisting of a dense cohesive core and a sparse, loosely connected periphery. This is due by the fact that the ML strategy increases the default probability for systematically important banks. Although these two strategies are simplified and do not consider the amount of capital that the LOLR has to inject in the banks network, nonetheless such comparison is useful because the numerical approach fits easily through simulations and systemic risk analysis.

Our work has some important similarities with the one by Capponi et al., in particular we also have considered a finite connected multi-period financial network representing the banks system and the assumptions guaranteeing the consistency with the bankruptcy laws. But, despite this, instead of comparing the two strategies, SID and ML, we considered a LOLR wishing to minimize the square of the lend resources over the probabilistic constraint. Therefore, we did not gave an initial budget at disposal to the LOLR as in \cite{Cap}, but took into consideration regulators aiming to find the loan control $\{\alpha^i(t)\}_{i=1,\dots,N,t\in[t_k,t_{k+1}]}$ minimizing the functional given by equation \eqref{EQN:CostFunc} for each time interval, i.e. $\forall k=1,\dots,M-1$, ensuring that the probability for each exogenous asset value to be greater than the default boundary is greater than a given constants $q^i$ for each bank $i\in\{1,\dots,N\}$.

Moreover, while \cite{Cap} is meant to compare two strategies for the LOLR, our approach follows a different path in searching the optimal budget consumption  to guarantee a prescribed level of safety of the financial network, given by  the parameters $q^i$ $i=1,\dots,N$. In particular, we do not assume strong constraint over the regulators budget, which depends on the default probability constraint parameters $q^i$. To switch on a similar comparison as in \cite{Cap}, i.e. considering banks networks of the type {\it core-periphery and baseline random networks}, and regulator policies of the type {\it SID and ML}, it suffices to fix the probability constraint depending on the systematic importance of the banks.
That is, banks whose failure would cause significant losses to the financial network, because of their size and systemic interconnectedness, should be endorsed with greater default probability parameters $q_i$. Therefore, our study provides an extension of the admissible policies, through considering an optimal control theory approach.

\end{document}